\documentclass[12pt]{article}
\usepackage{graphicx}
\usepackage{setspace}
\usepackage{empheq}
\usepackage{cite}
\usepackage{amsmath}
\usepackage{textcomp}
\usepackage{amssymb}
\usepackage{dutchcal}
\usepackage{authblk}
\usepackage{bm}

\usepackage{color}
\definecolor{MyLinkColor}{rgb}{0,0,0.4}

\newcommand{\h}{\mathcal{h}}
\newcommand{\kb}{\mathcal{k}}

\newcommand{\be}{\mathbf{e}}

\newtheorem{thm}{Theorem}[section]

\newtheorem{rem}[thm]{Remark}
\newenvironment{proof}{\paragraph{Proof:}}{\hfill$\square$}
\usepackage{fancyhdr}
\begin{document}
\title{Azimuthal equatorial flows in spherical coordinates with discontinuous stratification}

\author[$\star$]{Calin I. Martin}
\affil[$\star$]{Faculty of Mathematics, University of Vienna, Oskar-Morgenstern-Platz 1, 1090 Vienna, Austria, calin.martin@univie.ac.at}
\lhead{C. I. Martin}
\rhead{Azimuthal flows with discontinuous stratification}
\maketitle

\begin{abstract}
We are concerned here with an exact solution to the governing equations for geophysical fluid dynamics in spherical coordinates which incorporates discontinuous fluid stratification. This solution represents a steady, purely--azimuthal equatorial flow with an associated free-surface and an interface separating two fluid regions, each of which having its own continuous distribution of density. However, the two density functions do not match along the interface. Following the derivation of the  solution we demonstrate that there is a well-defined relationship between the imposed pressure at the free-surface and the resulting distortion of the surface's shape. Moreover, imposing the continuity of the pressure along the interface generates an equation that describes (implicitly) the shape of the interface. Interestingly, it turns out that the interface defining function has infinite regularity.
\end{abstract}
\bigbreak
\bigbreak

\section{Introduction}
The Pacific Ocean displays within a band of 150 km on each side of the Equator and extending longitudinally over about 16.000 km some outstanding features. First and foremost we would like to point out a pronounced density stratification
(greater than anywhere else in the ocean \cite{FedBr}) whose hallmark is the presence of a rather sharp interface (called pycnocline or thermocline) that separates a shallow near-surface layer of relatively warm water from a deep layer of colder and denser water.
In fact,  flow stratification (which in our study varies with depth) represents an intrinsic trait of geophysical water flows that particularly applies to large-scale ocean movements, cf. \cite{Gill,MP, Val}. Greatly impacted by changes in temperature and salinity, density fluctuations arise frequently in the ocean and lead to a layering of the flow: fluid layers of different densities organize themselves so that the higher densities are found below lower densities, cf. \cite{BecCushRoi, CJGAFD, CJPoF17, FedBr, Kes, McC}. This vertical layering ushers an expected bundle of properties which considerably influences the velocity field.  As a result internal gravity waves may appear, cf. Theorem \ref{existence}.

On the other hand, the Coriolis forces (created by the Earth's rotation) together with the prevailing westward wind pattern generate an underlying current field that displays flow-reversal \cite{Boy, CJPoF19}: the current field changes (in a band of about $2^{\circ}$ latitude around the Pacific Equator) from a westward flow near the surface to an eastward-flowing jet called the Equatorial Undercurrent (EUC) whose core resides more or less on the thermocline. An interesting aspect in the 
equatorial ocean dynamics is that the change of sign of the Coriolis force across the Equator turns the latter into an actual waveguide that enables azimuthal flow propagation which is symmetric with respect to the Equator.

Comprehensive mathematical models that handle wave-current interactions and internal waves in the presence of density stratification within the setting of nonlinear geophysical governing equations are quite rare and of very recent date \cite{ComAll, CJGAFD, CI, cim, CICMP, CulIva, IvaNARWA}. The stratification issue alone, in spite of being of utmost 
practical relevance, has remained unamenable to a rigorous mathematical analysis until recently:  for a selection of recent advances in the field (in the scenario of two-dimensional gravity water waves without Earth's rotation effects) we refer the reader to \cite{CI, cim, CICMP, EMM, Eschall, HM2, MatMH, HenMatAV, Wal1, Wal2}.

In this paper, we address the topic of density stratification of discontinuous type for flows exhibiting vertical structure, internal waves and a preferred propagation direction and so derive and, subsequently, analyze an exact and explicit solution which describes in a rotating frame a stratified inviscid and incompressible azimuthal equatorial water flow which presents a free surface and an interface. We would like to point out that explicit and/or exact solutions in fluid dynamics represent a very rare occurrence. They are important not only from a mathematical perspective but also for providing a basis for more
obtaining more elaborate and physically meaningful solutions by employing perturbative or asymptotic methods \cite{JohnPhT}.

Mindful of the previous aspects we aim at furthering a line of work put forward by Constantin \& Johnson concerning mathematical analyses of geophysical water flows.
First, Constantin \& Johnson constructed in \cite{CJaz, CJazAcc} in terms of spherical coordinates exact solutions to the full GFD governing equations. These solutions represent  purely-azimuthal, depth-varying flows, can be chosen to model both the equatorial undercurrent (EUC),  and the Antarctic Circumpolar Current (ACC), respectively, and describe purely homogeneous flows without stratification. Recently, Henry \& Martin presented equatorial flow solutions that allow for
continuous stratification that varies linearly with depth \cite{HenMarDPDE, HenMarJDE} as well as of most general type \cite{HenMarARMA, HenMarNONL}. Martin \& Quirchmayr \cite{MarQuirchJMP} devised solutions flows modelling the ACC 
exhibiting a general fluid stratification that varies both with depth and latitude. The special trait of the solution that we derive consists in its ability to accommodate a discontinuous density stratification that varies with depth.
More precisely, we allow a vertical layering as described in the beginning, with two layers of different, non-constant densities, where the denser layer sits below the less dense one (stable stratification): thus an interface arises which resembles the role played by the thermocline in the setting of observed equatorial flows.

Subsequent to the presentation of the equations of motion in Section \ref{presentation} we derive in Section \ref{solutions} explicit formulas for the velocity field and for the pressure function. We then exploit the dynamic boundary condition and find a relation between the imposed pressure on the surface and the resulting surface distortion. Moreover, from the balance of forces at the interface between the two layers, we derive another equation for the interface defining function. Although the two relations which define implicitly the free surface and the interface, respectively, appear in a convoluted form, they have the redeeming feature of being amenable to a functional analytic study by means of the implicit function theorem.
We present in Section \ref{analysis} results which show that our exact solution manifests acceptable properties: a growth in pressure along the surface leads to a decline in height of the free surface. Moreover, we also show that the interface between the two layers has infinite regularity.
\section{Equations of motion}\label{presentation}
The governing equations for geophysical fluid dynamics (GFD) are formulated in a spherical coordinate system which is fixed at a point on the Earth's surface as follows. We work in spherical coordinates, denoted $(r, \theta,\varphi)$, where $r$ is the distance from the centre of the earth, $\theta$ (with $0\leq \theta\leq \pi$) is the polar angle, 
and $\varphi$ (with $0\leq \varphi< 2\pi$) is the azimuthal angle, see Figure \ref{spherical_coord}. The location 
of the North and South poles are at $\theta=0,\pi$, respectively, while the Equator is situated at $\theta=\frac{\pi}{2}$. 
The unit vectors in this $(r,\theta, \varphi)$ system are $(\be_r, \be_{\theta}, \be_{\varphi})$, respectively, and the corresponding velocity components are $({\bf w}, {\bf v}, {\bf u})$;
$\be_{\varphi}$ points from West to East, and 
$\be_{\theta}$ points from North to South. A noteworthy aspect is that the fluid domain is stratified where the stratification is brought about by the changes in the (discontinuous) density. More precisely, denoting with $R\approx 6378$ km the Earth's radius and with $r_0>0, r_1$ some constants, we assume that the flow consists of a top layer 
$$D_{top}:=\{(r,\theta, \varphi): R+r_1+h(\theta, \varphi)\leq r\leq R+r_0+k(\theta,\varphi)\},$$
that lies above a bottom layer 
$$D_{bottom}:=\{(r,\theta, \varphi): R+d(\theta, \varphi)\leq r\leq R+r_1+h(\theta,\varphi)\}.$$
While above $d$ is a given function, $h$ and $k$ are unknowns of the problem and denote the interface and the free surface defining functions, respectively.

\begin{figure}[h]
\begin{center}
\includegraphics*[width=0.75\textwidth]{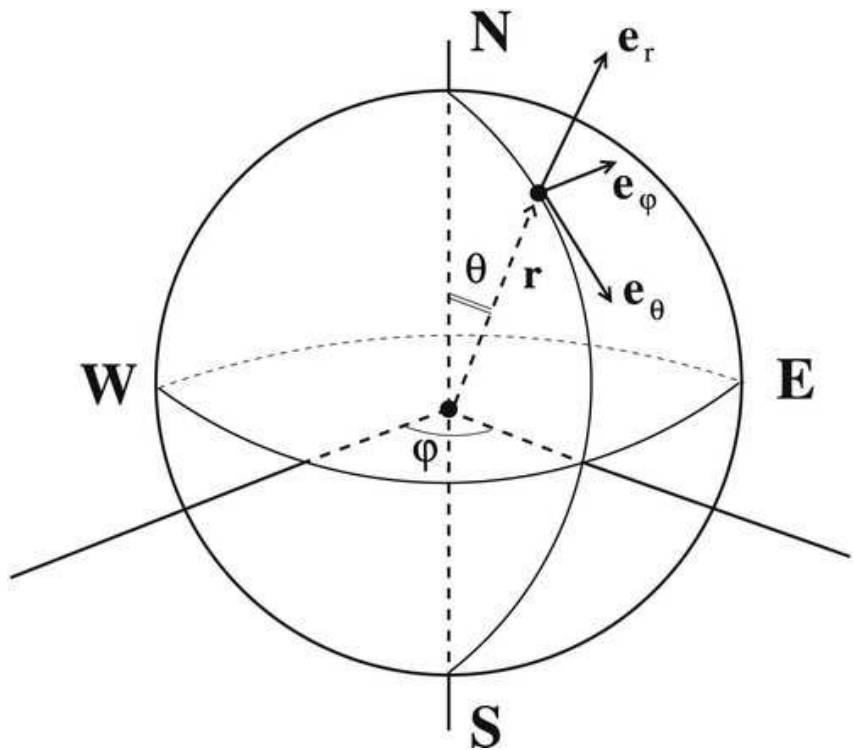}
\end{center}
\caption{The spherical coordinate system as presented in Section \ref{presentation}. The variable $r$ represents the distance from the origin, $\theta$ is the polar angle ($\pi/2-\theta$ being the angle of latitude), and $\varphi$ is the angle of longitude (azimuthal angle). } 
\label{spherical_coord}
\end{figure}

To draw attention to the layering aspect of the fluid domain we will use the index $1$ as a label for the upper layer. Whenever we refer to the overall physical variable without specification of the layer, we shall use bold face symbol. For example, 
the density $\bm{\rho}$ is assumed to be discontinuous across $r=R+h(\theta, \varphi)$ and is distributed in the two layers as 
\begin{equation}
\bm{\rho}(r)=\left\{ \begin{array}{l}\rho_1(r)\,\,{\rm for}\,\,(r,\theta,\varphi)\in D_{top},\\ 
 \rho(r)\,\,{\rm for}\,\,(r,\theta,\varphi)\in D_{bottom}, \end{array}\right. 
\end{equation}
being a differentiable function of $r$ in each of $D_{top}$ and $D_{bottom}$. Likewise, the components of the velocity field with respect to the orthonormal system $\{\be_r,\be_{\theta},\be_{\varphi}\}$ are denoted with $({\bf w}, {\bf v}, {\bf u})$. That is,
\begin{equation}
({\bf w}, {\bf v}, {\bf u})=\left\{\begin{array}{l}(w_1,v_1,u_1)\,\,{\rm in}\,\,D_{top},\\
(w,v,u)\,\,{\rm in}\,\,D_{bottom}.    \end{array}\right.
\end{equation}

\noindent We present now the equations of motion in a coordinate system with its origin at the centre of the sphere, where the effects of the earth's rotation are incorporated into our system. In order to do this we associate  $\{\be_r, \be_{\theta},\be_{\varphi}\}$ with a point fixed on the sphere which is rotating about its polar axis. Consequently,  the governing equations will contain Coriolis $2{\bm\Omega}\times \bf{U}$ terms and centripetal ${\bm\Omega}\times({\bm\Omega}\times \bf{r})$ terms,
with $ {\bf r}=r\be_r$, ${\bf U}={\bf w}\be_r+{\bf v}\be_{\theta}+{\bf u}\be_{\varphi}$, and
${\bm\Omega}=\varOmega(\be_r\cos\theta-\be_{\theta}\sin\theta),$
where $\varOmega\approx 7.29\times 10^{-5}$ rad s$^{-1}$ is the constant rate of rotation of the earth. The contribution of the Coriolis and centripetal acceleration
is therefore given by
\begin{equation*}
 2\varOmega(-{\bf u}\sin\theta\be_r-{\bf u}\cos\theta\be_{\theta}+({\bf w}\sin\theta +{\bf v}\cos\theta)\be_{\varphi})-r\varOmega^2(\sin^2\theta\be_r+\sin\theta\cos\theta\be_{\theta}).
\end{equation*}

Therefore, the GFD governing equations, in a spherical coordinate system with its origin at the centre of the sphere, are given by the Euler equation
 \begin{equation}\label{Eulereq'}
 \begin{split}
{\bf w}_t+{\bf w}{\bf w}_r +\frac{{\bf v}}{r}{\bf w}_{\theta}+\frac{{\bf u}}{r\sin\theta}{\bf w}_{\varphi}-\frac{1}{r}({\bf v}^2+{\bf u}^2) &-2\varOmega {\bf u}\sin\theta-r\varOmega^2 \sin^2\theta  \\
&= -\frac{{\bf p}_r}{\bm{ \rho}}+\mathfrak{F}_r\\
  \mathbf{ v}_t+{\bf w}{\bf v}_r +\frac{{\bf v}}{r}{\bf v}_{\theta}+\frac{{\bf u}}{r\sin\theta}{\bf v}_{\varphi}+\frac{1}{r}({\bf w}{\bf v}-{\bf u}^2\cot\theta) & -2\varOmega {\bf u}\cos\theta -r\varOmega^2 \sin\theta\cos\theta \\
  & = -\frac{ {\bf p}_{\theta}}{ \bm{\rho} r} +\mathfrak{F}_{\theta}\\
    {\bf u}_t+{\bf w} {\bf u}_r +\frac{ {\bf v} }{r}{\bf u}_{\theta}+\frac{ {\bf u} }{r\sin\theta}{\bf u}_{\varphi}+\frac{1}{r}({\bf w} {\bf u}+{\bf v}  {\bf u}\cot\theta)&+2\varOmega {\bf w}\sin\theta +2\varOmega {\bf v}\cos\theta\\
    &= -\frac{ {\bf p}_{\varphi}}{ \bm{\rho} r\sin\theta} +\mathfrak{F}_{\varphi},
\end{split}
\end{equation}
 and the equation of mass conservation
\begin{equation}\label{masscons}
 \frac{1}{r^2}\frac{\partial}{\partial r}(r^2  \bm{\rho} {\bf w})+\frac{1}{r\sin\theta}\frac{\partial}{\partial\theta}( \bm{ \rho} {\bf  v}\sin\theta)+\frac{1}{r\sin\theta}\frac{\partial }{\partial\varphi}(\bm{\rho} {\bf u})=0.
\end{equation}
Here ${\bf p}(r,\theta,\varphi)$ denotes the pressure in the fluid, $\mathfrak{F}=\mathfrak{F}_r\be_r+\mathfrak{F}_{\theta}\be_{\theta}+ \mathfrak{F}_{\varphi}\be_{\varphi}$ is the body-force vector.
The governing equations are supplemented by associated boundary conditions as follows. At the free-surface  we require the dynamic condition involving the surface pressure 
\begin{subequations}\label{Bound}
\begin{equation}\label{surfacepressure}
p_1=P_1(\theta,\varphi),
\end{equation}
and the kinematic condition 
\begin{equation}\label{kin_surf}
w_1=\frac{v_1}{r}\frac{\partial k}{\partial\theta}+\frac{u_1}{r\sin\theta}\frac{\partial k}{\partial\varphi},
\end{equation}
to hold. 

At the interface, $r=R+h(\theta, \varphi)$, we ask that the normal components of the velocity fields from the upper and lower layer, respectively, are the same. The latter is equivalent with the condition
\begin{equation}\label{kin_int}
\begin{split}
(w\be_r+v\be_{\theta}+u\be_{\varphi})&\cdot\left(\be_r-\frac{h_{\theta}}{r}\be_{\theta}-\frac{h_{\varphi}}{r\sin\theta}\be_{\varphi}\right)\\
&=(w_1\be_r+v_1\be_{\theta}+u_1\be_{\varphi})\cdot\left(\be_r-\frac{h_{\theta}}{r}\be_{\theta}-\frac{h_{\varphi}}{r\sin\theta}\be_{\varphi}\right).
\end{split}
\end{equation}
Moreover, to ensure the balance of forces at the interface, we also require that the pressure from the upper layer coincides with the pressure from the bottom layer, that is, 
\begin{equation}\label{balance_int}
p(R+h(\theta,\varphi),\theta,\varphi)=p_1(R+h(\theta,\varphi), \theta,\varphi).
\end{equation}

At the bottom of the ocean, which is an impermeable, solid boundary described by the equation $r=d(\theta,\varphi)$, the associated kinematic condition is
\begin{equation}\label{kin_bed}
 w=\frac{v}{r}\frac{\partial d}{\partial\theta}+\frac{u}{r\sin\theta}\frac{\partial d}{\partial\varphi}.
\end{equation}
\end{subequations}

\section{Explicit and exact solutions}\label{solutions}
We set out to derive exact solutions to the equations of motion and boundary conditions from Section \ref{presentation}.
In fact, the velocity field and the pressure function are given explicitly. However, due to intricacies like the usage of spherical coordinates, the complexity of stratification and the very involved formulas for the pressure and velocity, explicite formulas for the free surface and the interface remain elusive. An alleviation of this aspect is represented by a Bernoulli-type relation
between the imposed pressure at the surface and the resulting surface distorsion; an implicit formula for the interface defining function being given by the balance of forces at the interface \eqref{balance_int}. For a selective list of papers
presenting exact solutions pertaining to geophysical fluid dynamics we refer the reader to \cite{BasDcds, ChuIoneY, ACJGR, CoPhysOc2013,  CJaz, CJazAcc, CJPoF17, CJPRS, CJPoF19, HazMar, Hen2013, HsuJmfm14, HsuMh15, IoneJde20, Lyo, MarQuirchMfM, AMJPhA, AMApplAna, AMQAM, Quirch}.

Our objective is to find solutions to equations \eqref{Eulereq'}, \eqref{masscons} and \eqref{Bound} that represent purely-azimuthal steady flows with no variation in the azimuthal direction. Hence, the velocity field satisfies $w=v=w_1=v_1=0$ and $u=u(r,\theta),\,u_1=u_1(r,\theta)$. Moreover, the other unknows of the problem are characterized by the properties $p=p(r,\theta),\,p_1=p_1(r,\theta),\,h=h(\theta),\,k=k(\theta)$. We notice that a flow with the previous features automatically satisfies 
the boundary and interface conditions \eqref{Bound} as well as the equation of mass conservation \eqref{masscons}.
As for the Euler equations, they become equivalent with 
\begin{equation}\label{specialflow}
 \left\{\begin{array}{rcl}
         -\frac{{\bf u}^2}{r}-2\varOmega {\bf u}\sin\theta-r\varOmega^2\sin^2 \theta & =& -\frac{1}{\rho}p_r-g\\
         & &\\
          -\frac{{\bf u}^2}{r}\cot\theta- 2\varOmega {\bf u}\cos\theta -r\varOmega^2\sin\theta\cos\theta & = & -\frac{1}{\rho r}p_{\theta}\\
          &  &\\
0& =& -\frac{1}{\rho}\frac{1}{r\sin\theta}p_{\varphi}
          \end{array},\right. 
\end{equation}
where we assumed that the external body-force is due to gravity alone, giving the body-force vector as $-g\be_r$.
Eliminating the pressure from \eqref{specialflow} we obtain the equation
\begin{equation}\label{charac_met}
\partial_{\theta}\left(\frac{\rho(r)({\bf u}(r,\theta)+\Omega r\sin\theta)^2}{r}\right)-\partial_r\left(\rho(r)({\bf u}(r,\theta)+\Omega r\sin\theta)^2\cot\theta      \right)=0.
\end{equation}
To solve \eqref{charac_met} we appeal to the method of characteristics. An analysis
similar to the one in \cite{HenMarARMA} gives for the velocity field the formula
\\
\begin{equation}\label{vel_form}
{\bf u}(r,\theta)=\left\{
\begin{aligned}
u(r,\theta)=-\Omega r\sin\theta +\frac{F(r\sin\theta)}{\sqrt{\rho(r)}}\quad {\rm if}\quad R+d(\theta)\leq r\leq R_1+h(\theta),\\
u_1(r,\theta)=-\Omega r\sin\theta +\frac{F_1(r\sin\theta)}{\sqrt{\rho_1 (r)}}\quad {\rm if}\quad R_1+h(\theta)\leq r\leq R_0+k(\theta),
\end{aligned} \right.
\end{equation}
\\
where $R_0:=R+r_0, R_1:=R+r_1$, and $x\rightarrow F(x),\, x\rightarrow F_1(x)$ are arbitrary real-valued functions and $r_0,r_1$ are some positive constants.
Further, inserting the expressions of $w$ and $w_1$ from \eqref{vel_form} in \eqref{specialflow}  we have
\begin{equation}
{\bf p}_r=\left\{
\begin{aligned}p_r(r,\theta)=\frac{F^2(r\sin\theta)}{r}-g\rho(r)\quad {\rm if}\quad R+d(\theta)\leq r\leq R_1+h(\theta),\\
p_{1,r}(r,\theta)=\frac{F_1^2(r\sin\theta)}{r}-g\rho_1(r)\quad {\rm if}\quad R_1+h(\theta)\leq r\leq R_0+k(\theta),
\end{aligned}\right.
\end{equation}
and 
\begin{equation}
{\bf p}_{\theta}=\left\{
\begin{aligned}p_{\theta}(r,\theta)=F^2(r\sin\theta)\cot\theta\quad {\rm if}\quad R+d(\theta)\leq r\leq R_1+h(\theta),\\
p_{1,\theta}(r,\theta)=F_1^2(r\sin\theta)\cot\theta\quad {\rm if}\quad R_1+h(\theta)\leq r\leq R_0+k(\theta).
\end{aligned}\right.
\end{equation}
From integration with respect to $r$ we obtain for all $r\in  [R+d(\theta), R_1+h(\theta)]$

\begin{equation}
p(r,\theta)=\int_{(R+d(\theta))\sin\theta}^{r\sin\theta}\frac{F^2(y)}{y}dy-g\int_{R+d(\theta)}^r\rho(\tilde{r})d\tilde{r}
+f(\theta),
\end{equation}
where
\begin{equation}\label{f}
f(\theta)=\int_{R+d(\pi/2)}^{(R+d(\theta))\sin\theta}\frac{F^2(y)}{y}dy-g\int_{R+d(\pi/2)}^{R+d(\theta)}\rho(y)dy.
\end{equation}
For the upper layer we have 
\begin{equation}\label{pressure_up}
p_1(r,\theta)=\int_{(R_1+h(\theta))\sin\theta}^{r\sin\theta}\frac{F_1^2(y)}{y}dy-g\int_{R_1+h(\theta)}^r\rho_1 (\tilde{r})d\tilde{r}+c_1(h,\theta)
\end{equation}
 for all $r\in [R_1+h(\theta), R_0+k(\theta)]$, where 
 \begin{equation}\label{c1}
 \begin{aligned}
 c_1(h,\theta)=&\int_{\pi/2}^{\theta}F_1^2\big((R_1+h(\tilde{\theta}))\sin\tilde{\theta}\big)\left[\cot\tilde{\theta}+
 \frac{h'(\tilde{\theta})}{R_1+h(\tilde{\theta})}\right]\,d\tilde{\theta}\\
 &-g\int_{\pi/2}^{\theta}\rho_1(R_1+h(\tilde{\theta}))h'(\tilde{\theta})\,d\tilde{\theta}.
 \end{aligned}
 \end{equation}
\\ 
We pass now to the (implicit) determination of the two interfaces: the free surface and the interface separating the two layers $D_{bottom}$ and $D_{top}$.
From the dynamic condition on the surface \eqref{surfacepressure} we obtain 
\\
\begin{equation}\label{Bernoulli}
P_1(\theta)=\int_{(R_1+h(\theta))\sin\theta}^{(R_0+k(\theta))\sin\theta}\frac{F_1^2(y)}{y}dy-g\int_{R_1+h(\theta)}^{R_0+k(\theta)}\rho_1 (\tilde{r})d\tilde{r}+c_1(h,\theta).
\end{equation}
\\
Setting $h=0$ and $k=0$ in the previous formula we obtain the quantity
\\
\begin{equation}
P^0_1(\theta)=P_1(h\equiv 0)=\int_{R_1\sin\theta}^{R_0\sin\theta}\frac{F_1^2(y)}{y}dy-g\int_{R_1}^{R_0}\rho_1 (\tilde{r})d\tilde{r}+c_1(0,\theta),
\end{equation}
\\which stands for the pressure required to maintain an undisturbed free surface and an undisturbed interface following the curvature of the Earth.
\\
The balance of forces at the interface, encoded in the relation
 \begin{equation}\label{eq_int}
 p(R_1+h(\theta),\theta)=p_1(R_1+h(\theta),\theta),
 \end{equation}
 will allow us to implicitly determine the shape of the interface. Equation \eqref{eq_int} will be studied in a 
 functional analytic framework by means of nondimensionalization. More precisely, setting
 \begin{equation}
 \mathcal{h}(\theta):=\frac{h(\theta)}{R_1}
 \end{equation}
 we see that \eqref{eq_int} can be written as $\mathcal{G}(\mathcal{h})=0$, where 
 \\
 \begin{equation} \label{mathcalG}
 \begin{aligned}
 &\mathcal{G}(\h)(\theta)=\\
 &\frac{1}{P_{atm}}\left(\int_{(R+d(\theta))\sin\theta}^{(1+\h(\theta))R_1\sin\theta}\frac{F^2(y)}{y}dy-g\int_{R+d(\theta)}^{R_1(1+\h(\theta))}\rho(\tilde{r})d\tilde{r}
+f(\theta)\right)\\
&-\frac{1}{P_{atm}}\int_{\pi/2}^{\theta}F_1^2\big((1+\h(\tilde{\theta}))R_1\sin\tilde{\theta}\big)\left[\cot\tilde{\theta}+
 \frac{\h'(\tilde{\theta})}{1+\h(\tilde{\theta})}\right]\,d\tilde{\theta}\\
 &+\frac{gR_1}{P_{atm}}\int_{\pi/2}^{\theta}\rho_1\big(R_1(1+\h(\tilde{\theta}))\big)\h'(\tilde{\theta})\,d\tilde{\theta},
\end{aligned}
\end{equation}
\\
where $f(\theta)$ is given in \eqref{f}.
Also, the implicit determination of the free surface will be devised from the dynamic condition \eqref{surfacepressure}, the latter being written in the abstract form
\begin{equation}\label{abs_eq}
\mathcal{F}(\h, \mathcal{k} , \mathcal{P}_1)=0,
\end{equation}
with $\mathcal{F}$ acting from the Banach space
$C^1\big(\left(\left[\pi/2,\pi/2+\varepsilon\right]\right)\times C\left(\left[\pi/2,\pi/2+\varepsilon\right] \right)\times 
C\left(\left[\pi/2,\pi/2+\varepsilon\right]\right)\big)$ into itself, where $\varepsilon=0.016$ is suitable for flows in the equatorial region and demarcates a strip of about 100 km about the Equator, cf. \cite{CJaz}.
Above $\mathcal{k} (\theta):=\frac{k(\theta)}{R_0},\,\,\mathcal{P}_1(\theta):=\frac{P(\theta)}{P_{atm}}$
 and
\begin{equation}
\begin{split}
&\mathcal{F}(\h, \kb, \mathcal{P}_1)\\
&=\frac{1}{P_{atm}}\left(\int_{(1+h(\theta))R_1\sin\theta}^{(1+\kb(\theta))R_0\sin\theta}\frac{F_1^2(y)}{y}dy-g\int_{R_1(1+\h(\theta))}^{R_0(1+\kb(\theta))}\rho_1 (\tilde{r})d\tilde{r}+c_1(h,\theta) \right) -\mathcal{P}_1(\theta),
\end{split}
\end{equation}
where $c_1$ is given by \eqref{c1}. 
We can at once identify one trivial solution of \eqref{abs_eq}, namely the one associated with the flow having the free surface and the interface in an undisturbed state, following the curvature of the Earth. That is, we have
$$\mathcal{F}(0,0,\mathcal{P}_1^0)=0,$$ where
$$\mathcal{P}_1^0=\frac{P_1^0}{P_{atm}}=\frac{1}{P_{atm}}\left(\int_{R_1\sin\theta}^{R_0\sin\theta}\frac{F_1^2(y)}{y}dy-g\int_{R_1}^{R_0}\rho_1(\tilde{r})\,d\tilde{r}+
\int_{\pi/2}^{\theta}F_1^2(R_1\sin\tilde{\theta})\cot\tilde{\theta}\,d\tilde{\theta}\right).$$
We recast now the problem of finding $h$ and $k$ as the problem of finding solutions to the equation 
\begin{equation}\label{surf_int_syst}
(\mathcal{G}(h),\mathcal{F}(\h,\kb,P_1))=0.
\end{equation}
The previous task will be carried out by resorting to the implicit function theorem. Therefore, we will need to compute the 
derivative of the map $(\mathcal{G},\mathcal{F})$ with respect to $(h,k)$. To this end we compute first
\begin{equation*}
\begin{aligned}
\frac{c_1(s\h,\theta)-c_1(0,\theta)}{s}=&\frac{1}{s}\int_{\pi/2}^{\theta}\left[F_1^2\big( (R_1+sR_1\h(\tilde{\theta}))\sin\tilde{\theta}\big)-F_1^2(R_1\sin\tilde{\theta})\right]\cot\tilde{\theta}\,d\tilde{\theta}\\
&+\int_{\pi/2}^{\theta}\frac{\h'(\tilde{\theta})}{1+s\h(\tilde{\theta})}F_1^2\big( (R_1+sR_1\h(\tilde{\theta}))\sin\tilde{\theta}\big)\,d\tilde{\theta}\\
&-gR_1\int_{\pi/2}^{\theta} \rho_1(R_1+sR_1\h(\tilde{\theta}))\h'(\tilde{\theta})\,d\tilde{\theta}.
\end{aligned}
\end{equation*}
Utilizing the mean value theorem for integrals we obtain
\begin{equation}\label{c1_deriv}
\begin{aligned}
\lim_{s\to 0}\frac{c_1(sh,\theta)-c_1(0,\theta)}{s}=&\int_{\pi/2}^{\theta}\left[(F_1^2)'(R_1\sin\tilde{\theta}) R_1\h(\tilde{\theta})\cos\tilde{\theta}
+\h'(\tilde{\theta}) F_1^2\big( R_1\sin\tilde{\theta}\big) \right]\,d\tilde{\theta}\\
&-gR_1\rho_1(R_1)\big( \h(\theta)-\h(\pi/2)\big)\\
=&[F_1^2(R_1\sin\theta) -g R_1\rho_1(R_1)] \h(\theta),
\end{aligned}
\end{equation}
where, in the last step, we took into account that $\h(\pi/2)=0$.
Using \eqref{c1_deriv} we obtain
\begin{equation}
\begin{aligned}
P_{atm}\cdot (\mathcal{G}_{\h}(0)\h)(\theta)=&\lim_{s\to 0}\frac{\mathcal{G}(s\h)-\mathcal{G}(0)}{s}\\
=&\lim_{s\to 0}\frac{1}{s}\left[\int_{R_1\sin\theta}^{(R_1+sR_1\h(\theta))\sin\theta}\frac{F^2(y)}{y}\,dy-g\int_{R_1}^{R_1+sR_1\h(\theta)} \rho(\tilde{r})\,d\tilde{r}\right]\\
&-\lim_{s\to 0}\frac{1}{s}\int_{\pi/2}^{\theta}\left(F_1^2((R_1+sR_1\h(\tilde{\theta}))\sin\tilde{\theta})-F_1^2(R_1\sin\tilde{\theta})\right)\cot\tilde{\theta}\,d\tilde{\theta}\\
&-\lim_{s\to 0}\int_{\pi/2}^{\theta}F_1^2((R_1+sR_1\h(\tilde{\theta}))\sin\tilde{\theta})\frac{\h'(\tilde{\theta})}{1+s\h(\tilde{\theta})}\,d\tilde{\theta}\\
&+gR_1\lim_{s\to 0}\int_{\pi/2}^{\theta}\rho_1(R_1+sR_1\h(\tilde{\theta}))\h'(\tilde{\theta})\,d\tilde{\theta}.
\end{aligned}
\end{equation}
Appealing to the mean value theorem for integrals we obtain
\\
\begin{equation}
 (\mathcal{G}_h(0)h)(\theta)=\big[F^2(R_1\sin\theta)-F_1^2(R_1\sin\theta)-gR_1(\rho(R_1)-\rho_1(R_1))\big]\h(\theta).
\end{equation}
\\
\begin{rem}\label{Gh}
Field data suggest that the sizes of the densities $\rho$ and $\rho_1$ relate as $\rho=\rho_1 (1+\sigma)$, where
$\tau \rightarrow \sigma (\tau)$ is a positive function satisfying $\sigma=\mathcal{O}(10^{-3})$. Moreover, guided by physical reality, we also assume that $\rho_1=\mathcal{O}(10^{3})\, kg\cdot m^{-3}$. Accounting for the previous considerations we find that
\begin{equation}
\begin{aligned}
F^2(R_1\sin\theta)&-F_1^2(R_1\sin\theta)=\\
&=\rho(R_1)(u(R_1,\theta)+\Omega R_1\sin\theta)^2-\rho_1(R_1)(u_1(R_1,\theta)+\Omega R_1\sin\theta)^2\\
&=\rho_1 (R_1)(u^2-u_1^2+2\Omega R_1 (u-u_1)\sin\theta)+\rho_1\sigma (u+\Omega R_1\sin\theta)^2\\
&<\rho_1 u^2 +2\Omega R_1\rho_1 u\sin\theta +\rho_1\sigma (u+\Omega R_1\sin\theta)^2.\\
\end{aligned}
\end{equation}
\\
Allowing for values of $w,\rho_1$, that bear relevance on physical grounds, we have from above that 
$$F^2(R_1\sin\theta)-F_1^2(R_1\sin\theta)<R_1 \,kg\cdot m^{-1}\cdot s^{-2},$$
while
$$gR(\rho-\rho_1)=g R_1\rho_1\sigma>4.9 R_1\,kg\cdot m^{-1}\cdot s^{-2}.$$
The previous two inequalities show that there is a constant $\alpha<0$ such that
 $$F^2(R_1\sin\theta)-F_1^2(R_1\sin\theta)-gR_1(\rho(R_1)-\rho_1(R_1))\leq \alpha$$ holds true for all $\theta\in [0,\pi]$. The latter finding guarantees that the map $$\h\rightarrow\mathcal{G}_{\h}(0)\h$$ is a linear homeomorphism from 
$C^1\left(\left[\pi/2,\pi/2+\varepsilon\right]\right)$ to $C^1\left(\left[\pi/2,\pi/2+\varepsilon\right]\right)$.
\end{rem}
We further compute
\begin{equation}
\begin{aligned}
\mathcal{F}_{\h}(0,0,P_1^0)\h=&\lim_{s\to 0}\frac{\mathcal{F}(s\h,0,P_1^0)-\mathcal{F}(0,0,P_1^0)}{s}\\
=&\frac{-F_1^2(R_1\sin\theta)\h+gR_1\rho_1(R_1)\h+F_1^2(R_1\sin\theta)\h-gR_1\rho_1(R_1)\h}{P_{atm}}\\=&0.
\end{aligned}
\end{equation}
Finally, utilizing \eqref{vel_form}, we obtain that for all $\theta\in [0,\pi]$ it holds
\begin{equation}
\begin{aligned}
(\mathcal{F}_{\kb}(0,0,P_1^0)\kb)(\theta)=&\lim_{s\to 0}\frac{\mathcal{F}(0,s\kb,P_1^0)-\mathcal{F}(0,0,P_1^0)}{s}\\
=&\frac{F_1^2(R_0\sin\theta)\kb(\theta)-gR_0\rho_1(R_0)\kb(\theta)}{P_{atm}}\\
=&\frac{\rho_1(R_0)}{P_{atm}}  \left[(u(R_0,\theta)+\Omega R_0 \sin\theta)^2 -gR_0\right]\kb(\theta).
\end{aligned}
\end{equation}
\begin{rem}\label{Fk}
Arguing as in Remark \ref{Gh} we conclude that there is $\beta<0$ such that 
$u(R_0,\theta)+\Omega R_0 \sin\theta)^2 -gR_0\leq\beta$ for all $\theta\in [0,\pi]$. 
Consequently, we have that 
the map $$\kb\rightarrow\mathcal{F}_{\kb}(0,0,P_1^0)\kb$$ is a linear homeomorphism  from 
$C\left(\left[\pi/2,\pi/2+\varepsilon\right]\right)$ to $C\left(\left[\pi/2,\pi/2+\varepsilon\right]\right)$.
\end{rem}
\noindent Hence
\\
 $$(\mathcal{G},\mathcal{F})_{\h,\kb}(0,0,P_1^0)=
\begin{pmatrix}\mathcal{G}_{\h}(0,0,P_1^0) & \mathcal{G}_{\kb}(0,0,P_1^0)\\  \mathcal{F}_{\h}(0,0,P_1^0) & \mathcal{F}_{\kb}(0,0,P_1^0)\end{pmatrix}=
\begin{pmatrix}\mathcal{G}_{\h}(0,0,P_1^0) &0 \\ 0 &  \mathcal{F}_{\kb}(0,0,P_1^0)\end{pmatrix}
$$
\\
is a linear homeomorphism from 
$C^1\left(\left[\pi/2,\pi/2+\varepsilon\right]\right)\times C\left(\left[\pi/2,\pi/2+\varepsilon\right]\right)$ to \\
$C^1\left(\left[\pi/2,\pi/2+\varepsilon\right]\right)\times C\left(\left[\pi/2,\pi/2+\varepsilon\right]\right)$ 
by Remark \ref{Gh} and Remark \ref{Fk}. Hence, by the implicit function theorem, we have the following existence result concerning the free surface and the interface.
\begin{thm}\label{existence}
 For any sufficiently small perturbation
$\mathcal{P}_1$ of $\mathcal{P}_1^0$ there exists a unique tuple $(\h,\kb)\in C^1\left(\left[\pi/2,\pi/2+\varepsilon\right]\right)\times C\left(\left[\pi/2,\pi/2+\varepsilon\right]\right)$ such that \eqref{surf_int_syst} holds true.
\end{thm}
We would like to conclude this section by illustrating the usefulness of the exact solutions we found. More precisely, we will show that observed salient features of equatorial flows (like a pronounced density stratification, giving rise to the pycnocline,
as well as the presence of depth-dependent current fields) can be successfully captured by the setting provided in \eqref{vel_form}.
\begin{rem}
Firstly, we note that it is possible to determine the azimuthal flow velocity at latitude $\theta$ from its value at the Equator ($\theta=\pi/2$): indeed, it is a consequence of formula \eqref{vel_form} that 
\begin{equation}\label{equat_f}
u(r,\theta)=u(r,\pi/2)\sin\theta +\frac{1}{\sqrt{\rho(r)}}(F(r\sin\theta)-F(r)\sin\theta)\,\,\,{\rm in}
\,\,\,\Omega_{bottom},
\end{equation}
and
\begin{equation}\label{equat_f1}
u_1(r,\theta)=u_1(r,\pi/2)\sin\theta +\frac{1}{\sqrt{\rho_1(r)}}(F_1(r\sin\theta)-F_1(r)\sin\theta)\,\,\,{\rm in}
\,\,\,\Omega_{top}.
\end{equation}
Moreover, the functions $F$ and $F_1$ can be reconstituted from the equatorial  part of the flow as
\begin{equation}\label{FF1_recon}
F(r)=\sqrt{\rho(r)}(u(r,\pi/2)+\Omega r),\quad F_1(r)=\sqrt{\rho_1(r)}(u_1(r,\pi/2)+\Omega r).
\end{equation}
Next we make use of formulas \eqref{equat_f},\eqref{equat_f1} and \eqref{FF1_recon} to show that the explicit solutions derived here are relevant in the modeling of 
the remarkable features of equatorial flows. One such occurrence is the Equatorial Undercurrent (EUC), a current that runs 
throughout the extent of the equatorial Pacific. While the surface flow is--due to the prevailing trade-winds-- preponderantly westward, a flow reversal occurs beneath the surface leading to an eastward-flowing jet whose core resides more of less on the thermocline and which is confined to depths of $100-200$ m, cf. \cite{CJaz, Izu, JMcPE,Sir}. Deeper down, at
depths in excess of about $240$ m, there is, essentially, an abyssal layer of still water.

Let us denote with $u_w$ the magnitude of the westward surface velocity and with $u_e$ the maximal eastward velocity at the core of the EUC jet. A simple choice for $u(r,\pi/2)$ and $u_1(r,\pi/2)$ that accommodates the previously described traits of EUC is the parabolic profile
\begin{equation}
u_1(r,\pi/2)=u_e-(u_e+u_w)\left(\frac{r-R_1}{R_0-R_1}\right)^2\,\,{\rm for}\,\,R_1\leq r\leq R_0,
\end{equation}
and 
\begin{equation}
u(r,\pi/2)=\left\{\begin{array}{ll}
 u_e-(u_e+u_w)\left(\frac{r-R_1}{R_0-R_1}\right)^2 &{\rm for}\quad\overline{R}\leq r\leq R_1\\
 0 &{\rm for}\quad r\leq\overline{R}
 \end{array}.\right.
\end{equation}
\\
Then, setting
\begin{equation}
F_1(r)=\sqrt{\rho_1(r)}\left[\Omega r +u_e-(u_e+u_w)\left(\frac{r-R_1}{R_0-R_1}\right)^2\right]\,\,{\rm for}\,\,R_1\leq r\leq R_0,
\end{equation}
and 
\begin{equation}
F(r)=\left\{\begin{array}{ll}
\sqrt{\rho(r)}\left[\Omega r+ u_e-(u_e+u_w)\left(\frac{r-R_1}{R_0-R_1}\right)^2\right] &{\rm for}\quad\overline{R}\leq r\leq R_1,\\
 \sqrt{\rho(r)}\Omega r &{\rm for}\quad r\leq\overline{R}
 \end{array}\right.
\end{equation}
leads to an elementary model of EUC. Here $r=R_0$ denotes the free surface at the Equator, $r=R_1$ is the depth at which EUC attains its maximum speed (typically this coincides with the thermocline) and 
$$\overline{R}:=R_1-(R_0-R_1)\sqrt{\frac{u_e}{u_e+u_w}}$$ represents the depth below which the fluid motion ceases, cf. Figure $\ref{euc}$.
\end{rem}

\begin{figure}[h]
\begin{center}
\includegraphics*[width=0.90\textwidth]{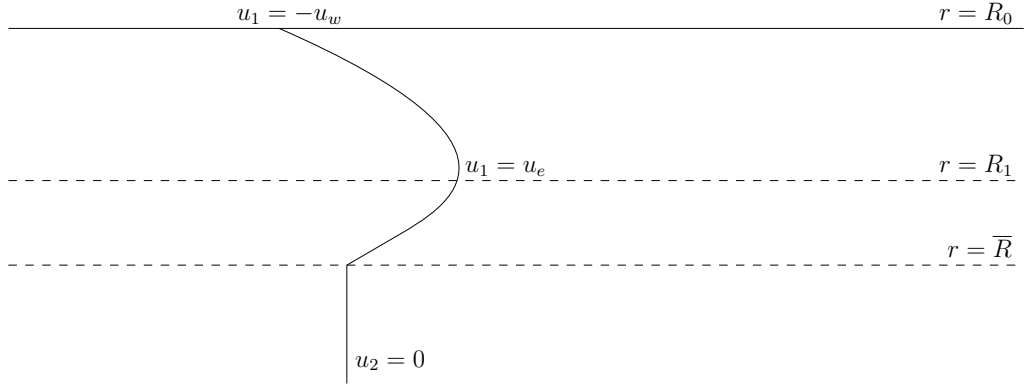}
\end{center}
\caption{A schematic displaying the salient features of typical equatorial current profiles in a band of about $2^{\circ}$ latitude around the Pacific equator: a significant fluid stratification, resulting in a pycnocline (also called thermocline) and a depth dependent current that changes from a westward flow near the surface to a quite significant flow to the east, located below the surface. The fluid motion dies out at depths beneath 240 m.} 
\label{euc}
\end{figure}

\section{Analysis of the exact solutions}\label{analysis}
This section is concerned with an in-depth analysis of the exact solutions derived previously. At first, we would like to reinforce the relevance of the cyclostrophic balance \eqref{Bernoulli}, which established the connection between the surface defining function and the pressure applied on the free surface. The next theorem allows us to reasonably conclude that our exact solution in a rotating frame with surface boundary conditions displays generally acceptable and expected physical properties. More precisely, it will be demonstrated that a growth in pressure generates a decline in height of the free surface away from the Equator.
\begin{thm}[Monotonicity properties]\label{mon}
Let $\mathcal{k}, \mathcal{P}_1$ denote the shape of the free surface and the pressure applied on the free surface, respectively.
If $\theta_0\in \left(\frac{\pi}{2},\frac{\pi}{2}+\varepsilon\right)$  denotes the polar angle of some location away from the Equator, then
\begin{equation}\label{h_beta}
\mathcal{P}_1'(\theta)<0\quad\textrm{if}\quad\mathcal{k}'(\theta)\geq 0,
\end{equation}
and
\begin{equation}\label{beta_h}
\mathcal{k}'(\theta)<0\quad\textrm{if}\quad\mathcal{P}_1'(\theta)\geq 0.
\end{equation} 
\end{thm}
\begin{proof}
To begin with, we notice that an iterative bootstrapping procedure \cite{Ber} allows us to transfer smoothing properties from $\mathcal{P}_1$ to $\mathcal{k}$.
By a differentiation with respect to $\theta$ in \eqref{abs_eq} we obtain
\begin{equation}
\begin{split}
P_{atm}\mathcal{P}_1'(\theta)=&\frac{F_1^2\left((R_0+k(\theta))\sin\theta\right)}{(R_0+k(\theta))\sin\theta}
\left(k'(\theta)\sin\theta+(R_0+k(\theta))\cos\theta\right)\\
\\
&-\frac{F_1^2\left((R_1+h(\theta))\sin\theta\right)}{(R_1+h(\theta))\sin\theta}
\left(h'(\theta)\sin\theta+(R_1+h(\theta))\cos\theta\right)\\
\\
&-g\left(\rho_1(R_0+k(\theta))k'(\theta)-\rho_1(R_1+h(\theta))h'(\theta)\right)\\
\\
&+\frac{F_1^2\left((R_1+h(\theta))\sin\theta\right)h'(\theta)}{R_1+h(\theta)}+F_1^2\left((R_1+h(\theta))\sin\theta\right)\cot\theta\\
\\
&-g\rho_1(R_1+h(\theta))h'(\theta).
\end{split}
\end{equation}
Hence, due to cancellations in the above formula, we obtain
\begin{equation}\label{dep_deriv}
\begin{split}
P_{atm}\mathcal{P}_1'(\theta)=&\left(\frac{F_1^2\left((R_0+k(\theta))\sin\theta\right)}{R_0+k(\theta)}-g\rho_1(R_0+k(\theta))\right)k'(\theta)\\
\\
&+F_1^2\left((R_0+k(\theta))\sin\theta\right)\cot\theta\\
\\
=& \left(\frac{\big[u_1(R_0+k(\theta),\theta)+\Omega (R_0+k(\theta) )\sin\theta \big]^2}{ R_0+k(\theta)  }-g\right) \rho_1(R_0+k(\theta)) k'(\theta)\\
\\
&+F_1^2\left((R_0+k(\theta))\sin\theta\right)\cot\theta\\
\\
=&\left(\frac{\big[u_1(R_0+k(\theta),\theta)+\Omega (R_0+k(\theta) )\sin\theta \big]^2}{ 1+\mathcal{k}(\theta)  }-gR_0\right) \rho_1(R_0+k(\theta)) \mathcal{k}'(\theta)\\
\\
&+F_1^2\left((R_0+k(\theta))\sin\theta\right)\cot\theta.
\end{split}
\end{equation}
Employing a similar type of argument as in Remark \ref{Gh} and Remark \ref{Fk} we notice that 
\\
$$
\frac{\big[u_1(R_0+k(\theta),\theta)+\Omega (R_0+k(\theta) )\sin\theta \big]^2}{ 1+\mathcal{k}(\theta)  }-gR_0<0,
$$ 
\\
for velocities $u_1$ which occur in realistic physical scenarios. The latter inequality and relation \eqref{dep_deriv} prove 
\eqref{h_beta} and \eqref{beta_h}.
\end{proof}

\noindent Another interesting and relevant upshot of the solutions proved to exist by Theorem \ref{existence} refers to the regularity of the interface defining function.
\begin{thm}[Regularity of the interface]\label{reg}
Let us assume that the distributions of density $\rho$ and $\rho_1$ are related by the equality $\rho=\rho_1(1+\sigma)$, where $\tau\rightarrow\sigma(\tau)$ is a positive function satisfying $\sigma=\mathcal{O}(10^{-3})$, cf. \cite{Kes}.
Then, the interface defining function $h$ belongs to $C^{\infty}\left[\pi/2,\pi/2+\varepsilon\right]$.
\end{thm}
\begin{proof}
We invoke equation \eqref{mathcalG} and
differentiate with respect to $\theta$ in $$\mathcal{G}(\h)(\theta)=0\,\,{\rm for}\,\,{\rm all}\,\,\theta,$$
obtaining
\\
\begin{equation}\label{eq_hprime}
\begin{aligned}
&\Big[\frac{F^2((R_1+R_1\h(\theta))\sin\theta)-F_1^2((R_1+R_1\h(\theta))\sin\theta)}{1+\h(\theta)}  \\
&-gR_1\big( \rho(R_1+h(\theta))-\rho_1(R_1+h(\theta))\big)\Big]\h'(\theta)\\
&+\left[F^2((R_1+R_1\h(\theta))\sin\theta)-F_1^2((R_1+R_1\h(\theta))\sin\theta)\right]\cot\theta =0.
\end{aligned}
\end{equation}
\\
We are going to prove that the bracket that multiplies $\h'(\theta)$ in the above formula is strictly positive. To this end we notice first that 
\\
\begin{equation}\label{FF_1}
\begin{aligned}
&\frac{F^2((R_1+R_1\h(\theta))\sin\theta)- F_1^2((R_1+R_1\h(\theta))\sin\theta)}{1+\h(\theta)}=\\
&\hspace{-3cm}\rho(R_1+R_1\h)\frac{\big( u(R_1+R_1\h,\theta)+\Omega R_1 (1+\h)\sin\theta\big)^2}{1+\h}\\
&\hspace{-3.3cm}-\rho_1(R_1+R_1\h)\frac{\big( u_1(R_1+R_1\h,\theta)+\Omega R_1(1+\h)\sin\theta\big)^2}{1+\h}.
\end{aligned}
\end{equation}
\\
Utilizing the relations between $\rho$ and $\rho_1$, as in the statement of the Theorem, we see that
 the right hand side of \eqref{FF_1} equals
\begin{equation}\label{rhs_term}
\begin{aligned}
&\frac{\rho_1}{1+\h}(u^2-u_1^2+2\Omega R_1 u(1+\h)\sin\theta-2\Omega u_1 R_1(1+\h)\sin\theta)\\
&+\rho_1\sigma \frac{ (u+\Omega R_1(1+\h)\sin\theta)^2}{1+\h}\\
&< \frac{\rho_1 u^2}{1+\h}+2\Omega R_1\rho_1 u\sin\theta +\rho_1\sigma \frac{ (u+\Omega R_1(1+\h)\sin\theta)^2}{1+\h}\\
&<R_1\left[\frac{\rho_1 u^2}{R_1}+2\Omega\rho_1 u\sin\theta +\rho_1\sigma \frac{ (u+\Omega R_1(1+\h)\sin\theta)^2}{R_1}\right]<R_1 \,kg\cdot m^{-1}\cdot s^{-2},
\end{aligned}
\end{equation}
where, for the last inequality, we took into account the realistic physical values of the involved quantities. On the other hand,
\begin{equation}\label{g_term}
gR_1(\rho-\rho_1)=g R_1\rho_1\sigma>4.9R_1\,kg\cdot m^{-1}\cdot s^{-2}.
\end{equation}
Owing to \eqref{rhs_term} and \eqref{g_term} we conclude that, indeed, the coefficient of $\h'(\theta)$ in 
\eqref{eq_hprime} is strictly negative. We now see that, since already $\h\in C^1 \left[\pi/2,\pi/2+\varepsilon\right]$, we
can differentiate in \eqref{eq_hprime} and infer that, in fact, $\h\in C^2\left[\pi/2,\pi/2+\varepsilon\right]$. An iteration of the previous proves the claim about the infinite differentiability of $h$.
\end{proof}

\begin{rem}
We present now an explicit solution (in terms of the velocity field, the pressure and the free surface and the interface)
for an apriori given density. More precisely, the scenario that we consider assumes that the density $\bm{\rho}$ is as follows: throughout the bottom layer
$\{(r,\theta,\varphi): R+d(\theta)\leq r\leq R_1+h(\theta)\}$ the density is constant and equals $\rho=\mathcal{O}(10^3)\,kg\cdot m^{-3}$. For the upper layer
$\{(r,\theta,\varphi): R_1+h(\theta)\leq r\leq R_0+k(\theta)\}$ we set 
\begin{equation}\label{dens1_expl}
\rho_1(r)=\rho-a_1 r.
\end{equation}
We begin by specifying the velocity field, that is
we particularize now the functions $F$ and $F_1$ appearing in the formula \eqref{vel_form} of the azimuthal velocity $\mathbf{u}$. We will then solve a differential equation for the interface defining function $h$, with the help of which we will find, via \eqref{Bernoulli}, the formula for the free surface defining function $k$.
More precisely, we set
 \begin{equation}\label{FF1_expl}
 F(x)=\Omega \sqrt{\rho} x\,\,{\rm and}\,\,F_1(x)=0\,\,{\rm for}\,\,{\rm all}\,\, x>0.
 \end{equation}
 Theferore, the azimuthal velocity is, cf. \eqref{vel_form}
 \begin{equation}\label{vel_exampl}
{\bf u}(r,\theta)=\left\{
\begin{aligned}
&u(r,\theta)=0\quad {\rm if}\quad R+d(\theta)\leq r\leq R_1+h(\theta),\\
&u_1(r,\theta)=-\Omega r\sin\theta \quad {\rm if}\quad R_1+h(\theta)\leq r\leq R_0+k(\theta).
\end{aligned} \right.
\end{equation}
\\
We will employ now \eqref{eq_hprime} to derive the equation for the function $h(\theta)$ defining the interface.
By means of the previous assumptions we obtain first that 
$$F^2((R_1+R_1\h(\theta))\sin\theta)=\rho\Omega^2 R_1^2(1+\mathcal{h}(\theta))^2\sin^2\theta$$
\\
and 
$$\rho(R_1+h(\theta))-\rho_1(R_1+h(\theta))=a_1 (R_1+h(\theta))=a_1 R_1 (1+\mathcal{h}(\theta)).$$
\\
Remaining consistent with the proportion between $\rho$ and $\rho_1$ proposed in Theorem \ref{reg} we choose 
$a_1:=\frac{1}{R_1}\cdot 2\, kg\cdot m^{-3}$.  Availing of the previous considerations, we obtain from \eqref{eq_hprime} that $\mathcal{h}$ satisfies the equation
\begin{equation}
\left[\rho \Omega^2\sin^2\theta -g a_1\right]\mathcal{h}'(\theta)+\rho \Omega^2  (1+\mathcal{h}(\theta))\sin\theta\cos\theta=0.
\end{equation}
\\ Using that $h(\pi /2)=0$ we solve the previous differential equation and obtain
\\
\begin{equation}
1+\mathcal{h}(\theta)=\sqrt{\frac{\rho\Omega^2-ga_1}  {\rho\Omega^2\sin^2\theta-ga_1}}.
\end{equation}
Taking into account that $\rho\Omega^2\approx 0.53\cdot 10^{-5} kg\cdot m^{-3}\cdot s^{-2}$ and $ga_1\approx  0.31\cdot 10^{-5}kg\cdot m^{-3}\cdot s^{-2}$ the equation of the interface is 
\\
\begin{equation}
r=R_0+h(\theta)=R_0\sqrt{ \frac{0.22}  {0.53\sin^2\theta-0.31}}.
\end{equation}
The previous formula, forecasts a rise in the height of the interface from $R_0$ at the Equator to $10^{-5}R_0$ at about $20$ km south and north of Equator, cf. the graph in Figure $\ref{interface}$.
\\
\begin{figure}
\begin{center}
\includegraphics*[width=0.85\textwidth]{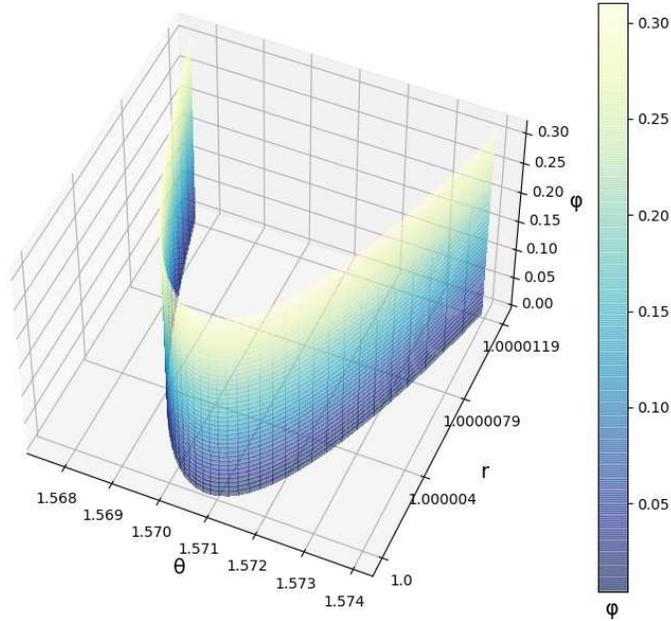}
\caption{The graph of the nondimensional function $1+\mathcal{h}(\theta)$. The range of the polar angle $\theta$ is the interval $[\frac{\pi}{2}-0.003141, \frac{\pi}{2}+0.003141]$. This corresponds in, physical variables, to strips of about 20 km width on each side of the Equator.} 
\label{interface}
\end{center}
\end{figure}

We now proceed to the explicit determination of the pressure in the upper layer $\Omega_{top}$. It will turn out that, for the special choices \eqref{dens1_expl} and \eqref{FF1_expl} we can determine an explicit formula for the pressure in the upper layer too. To this end, we note first that,
since we chose $F_1=0$, we obtain from \eqref{pressure_up} and \eqref{c1} the formula
\\
\begin{equation}\label{expl_interface}
\begin{aligned}
p_1(r,\theta)=&-g\int_{R_1+h(\theta))}^r \rho_1(\tilde{r})d\tilde{r}-g\int_{\pi/2}^{\theta}\rho_1 (R_1+h(\tilde{\theta}))
h'(\tilde{\theta})d\tilde{\theta}\\
=&-g\int_{R_1+h(\theta))}^r (\rho-a_1\tilde{r})d\tilde{r}-g\int_{R_1}^{R_1+h(\theta)}(\rho-a_1 z)dz\\
=&-\rho g (r-R_1)+\rho g h(\theta)+\frac{a_1 g}{2}[r^2-(R_1+h(\theta))^2]\\
&-\rho g h(\theta)+\frac{a_1 g}{2}(R_1+h(\theta))^2-\frac{a_1 g R_1^2}{2}\\
=&-\rho g (r-R_1)+\frac{a_1 g}{2}(r^2-R_1^2).
\end{aligned}
\end{equation}
\\
Now, the Bernoulli relation \eqref{Bernoulli} allows us to obtain the formula for the free surface defining function  $k(\theta)$. More precisely,  imposing a pressure $P_1(\theta)$ on the free surface, the surface defining function $\theta\rightarrow k(\theta)$ satisfies the equation
\\
\begin{equation}\label{eq_free_surf}
\frac{g a_1}{2}(R_0+k(\theta))^2-g\rho (R_0+k(\theta))+g\rho R_1 -\frac{g a_1 R_1^2}{2}-P_1(\theta)=0\,\,{\rm for}\,\,{\rm all}\,\,\theta.
\end{equation}
\\
We treat the previous equation as a second degree equation in the unknown $R_0+k(\theta)$ whose discriminant equals
$g^2(\rho-a_1 R_1)^2 +2ga_1 P_1(\theta)$.
Taking into account the reasonable physical size range for $R_0+k(\theta)$ we infer that 
\begin{equation}\label{expl_surface}
R_0+k(\theta)=\frac{g\rho -\sqrt{g^2(\rho-a_1 R_1)^2 +2ga_1 P_1(\theta)} }{g a_1}
\end{equation}
is the only possible solution to equation \eqref{eq_free_surf}.
\end{rem}

\section{Conclusions}
The analysis performed here provided explicit and exact steady solutions in spherical coordinates describing flows with a free surface and an interface which separates fluid regions of discontinuous densities. These flows manifest an azimuthal propagation direction with no variation in the azimuthal direction. Nevertheless, the azimuthal component of the velocity accommodates any arbitrary variation of speed with depth below the surface: indeed the functions $F$ and $F_1$ appearing in formula \eqref{vel_form},  can be chosen arbitrary. For any such choice, $w$ and $w_1$ from \eqref{vel_form}
provide a solution of the Euler equations \eqref{specialflow}. The usefulness of these solutions is rendered by their ability to integrate observed prominent features of remarkable oceanic flows like EUC: evidence in support of this assertion 
was provided in Remark 3.4.

The results in this paper represent an advancement toward the analytical understanding of geophysical water flows. Indeed,
allowing for a discontinuous density (that gives rise to an interface) was not addresed in previous papers \cite{CJaz, CJazAcc, CJPoF17, HenMarJDE, HenMarARMA} which handled geophysical flows with a preferred propagation direction.

The constraints that we impose on the solutions refer to the continuity of the pressure along the interface and the requirement that the pressure on the surface matches some given, but otherwise, arbitrary function. The latter leads to the 
Bernoulli relation on the free surface and shows that variations in pressure give rise to changes in the shape of the free surface. While the azimuthal velocity and the pressure function in the lower layer $\Omega_{bottom}$ are given explicitly, the pressure in the top layer $\Omega_{top}$ and the interface and surface defining functions, $h$ and $k$, respectively, appear (in general) in implicit form, cf. equations \eqref{Bernoulli} and \eqref{eq_int}. However, for some specific choice of the density function and of the azimuthal velocity, cf. \eqref{FF1_expl} and \eqref{vel_exampl}, we were able to provide a fully explicit solution describing an azimuthal flow: that is, also the interface and the free surface were given explicitly, cf. \eqref{expl_interface}, \eqref{expl_surface}.

One aspect that needs to be considered in future works refers to the 
lack of explicit formulas for the surface and interface defining functions in the general case. This drawback was, to some degree,
alleviated in Section \ref{analysis} by providing qualitative properties that reinforce the relevance of our exact solution. More precisely, we showed that a growth in the pressure along the free surface gives rise to a decline in height of the latter away from the Equator. We have also proved that the interface has infinite regularity.

Future endeavours refer to the need to address some other type of solutions which allow for more general velocity fields (e. g. having two or all three components non-vanishing) which are relevant in the study of other realistic ocean flows. For instance, one might want to allow a slow evolution in the azimuthal direction; this additional dependence on the slow variable will then yield a three-dimensional flow.\\
\\
\noindent {\bf Acknowledgements}: The author gratefully acknowledges the support of the Austrian Science Fund (FWF) through research grant P 33107-N. Comments and suggestions from anonymous referees are gratefully acknowledged. The author would like to thank Prof. B. Basu (School of Engineering, Trinity College Dublin) for the technical support in producing Figure 3.

\noindent {\bf Data availability}: The data that support the findings of this study are available from the corresponding author upon reasonable request.

\end{document}